\documentclass[english]{article}
\usepackage{amsfonts,amsmath,amsthm,amscd,amssymb,latexsym}
\usepackage{tikz}

\theoremstyle{plain}
\newtheorem{thm}{Theorem}[section]

\newtheorem{prop}[thm]{Proposition}
\newtheorem{rem}[thm]{Remark}
\theoremstyle{definition}
\theoremstyle{remark}
\numberwithin{equation}{section}
\newcommand{\keywords}{\textbf{Key words and phrases: }\medskip}
\newcommand{\subjclass}{\textbf{Math. Subj. Clas.: }\medskip}
\begin{document}
\title{\textbf{On a semi-discrete model of Maxwell’s equations in three and two dimensions} }
\author{\textbf{Volodymyr Sushch} \\
{ \em Koszalin University of Technology} \\
 {\em Sniadeckich 2, 75-453 Koszalin, Poland} \\
 { \em volodymyr.sushch@tu.koszalin.pl} }

\date{}
\maketitle
\begin{abstract}

In this paper, we develop a geometric, structure-preserving semi-discrete formulation of Maxwell’s equations in both three- and two-dimensional settings within the framework of discrete exterior calculus. This approach preserves the intrinsic geometric and topological structures of the continuous theory while providing a consistent spatial discretization. We analyze the essential properties of the proposed semi-discrete model and compare them with those of the classical Maxwell's equations. As a special case, the model is illustrated on a combinatorial two-dimensional torus, where the semi-discrete Maxwell's equations take the form of a system of first-order linear ordinary differential equations. An explicit expression for the general solution of this system is also derived.
\end{abstract}

\keywords{Maxwell's equations, discrete exterior calculus, discrete operators, combinatorial torus,
difference-differential equations}

 \subjclass  {39A12, 39A70,  35Q61}

 \section{Introduction}

The construction of discrete models that preserve the geometric structure of mathematical physics problems is fundamental to achieving reliable and physically consistent numerical simulations of differential equations. The present study continues our series of works \cite{S1, S2, S3, S4, S5,S6} in which discrete analogues of several fundamental equations of mathematical physics were developed using a geometric discretization framework based on discrete exterior calculus. The main idea of this approach originates from the work of Dezin \cite{Dezin}. In this paper, we introduce a discrete–continuous counterpart of Maxwell’s equations, where the spatial variables are discretized while the time variable remains continuous. The resulting   semi-discrete model is  represented by a system of first-order linear ordinary differential equations. We develop discrete versions of Maxwell’s equations in both three- and two-dimensional spatial settings with time dependence.

Numerous studies have addressed the problem of discretizing electromagnetic theory within the framework of the exterior calculus of differential forms (see, for example, \cite{AK, Bossavit1, Bossavit, ChCh, Christiansen, Hiptmair, Monk, Teixeira1, Teixeira}, and the references therein). Some of these approaches are based on lattice discretization schemes \cite{ChCh,Teixeira1, Teixeira}. Formulating Maxwell’s equations in the language of differential forms \cite{Deschamps} and employing discrete exterior calculus as the computational foundation have led to significant advancements in numerical methods based on finite element and finite difference techniques \cite{ACHZ, AK, Berchenko, Bossavit1, Bossavit, Kim, Monk}.  Numerical computations in the finite element exterior calculus method \cite{Arnold} are typically based on Whitney forms. The discretization scheme considered in the present paper, however, does not employ Whitney forms, nor does it use the Whitney or de Rham maps between cochains and differential forms \cite{Teixeira}. Nevertheless, the essential structure of exterior calculus is preserved in the discrete setting.

Let us briefly recall the key definitions involved in the standard three-dimensional formulation of Maxwell's equations using the framework of exterior calculus. See, for example, \cite{WR} or \cite{WR1} for details. In this formalism, electromagnetic fields and source quantities are described using differential forms:
The 1-forms
$E$ and $H$ represent the electric and magnetic field intensities, respectively. The 2-forms
$D$ and $B$ correspond to the electric and magnetic flux densities. The 2-form
$J$ denotes the electric current density, and finally, the 3-form
$Q$ represents the electric charge density.
Maxwell's equations can then be written as:
\begin{equation}\label{1.1}
 dE=-\frac{\partial B}{\partial t},
 \end{equation}
 \begin{equation}\label{1.2}
 dH=\frac{\partial D}{\partial t}+J,
 \end{equation}
 \begin{equation}\label{1.3}
 dD=Q,
 \end{equation}
 \begin{equation}\label{1.4}
 dB=0,
 \end{equation}
 where $d$ denotes the exterior derivative. The   constitutive relationships are given by
 \begin{equation}\label{1.5}
 D=\varepsilon_0\ast E,
 \end{equation}
 \begin{equation}\label{1.6}
 B=\mu_0\ast H,
 \end{equation}
 where $\varepsilon_0$  and $\mu_0$ are the vacuum permittivity and permeability, respectively, and $\ast$ denotes  the Hodge star acting in $\mathbb{R}^3$. In three dimensions, the Hodge star satisfies $\ast\ast = \mathrm{Id}$ for any forms.
Therefore, equations \eqref{1.5} and \eqref{1.6} can be equivalently written as
\begin{equation}\label{1.7}
 \ast D=\varepsilon_0 E,
 \end{equation}
 \begin{equation}\label{1.8}
 \ast B=\mu_0H.
 \end{equation}
Poynting's theorem, within the framework of differential forms, can be expressed as
 \begin{equation}\label{1.9}
d(E\wedge H)=-\frac{1}{2}\frac{\partial}{\partial t}(E\wedge D+B\wedge H)-E\wedge J,
\end{equation}
where $E\wedge H$ is the  Poynting energy flow form,
 $\frac{1}{2}E\wedge D$ and $\frac{1}{2}B\wedge H$ represent the electric and the magnetic densities, respectively, and $E\wedge J$ denotes the power density. See  \cite{WR} for details.
  In the two-dimensional case, Maxwell’s equations retain the form of Equations \eqref{1.1}–\eqref{1.6}, differing only in the interpretation of the Hodge star operator, which depends on the dimension \cite{WR1}.

The aim of this work is to develop a geometric structure-preserving semi-discrete formulation of Maxwell’s equations in both three- and two-dimensional settings. Building upon our previous studies \cite{S5, S6}, we construct discrete analogues of Equations \eqref{1.3}–\eqref{1.6} on a model of the two-dimensional torus. In this framework, the original system of partial differential equations is transformed into a system of linear ordinary differential equations that can be solved analytically.

The rest of the paper is organized as follows. In Section 2, we describe the construction of a combinatorial model of
$\mathbb{R}^3$, extending the combinatorial model of $\mathbb{R}^2$  presented in \cite{S5}. We introduce a cochain complex and define discrete analogues of the fundamental operations of exterior calculus. In Section 3, we establish a three-dimensional discrete counterpart of Maxwell’s equations while keeping time as a continuous variable. Furthermore, we examine the essential properties of the proposed semi-discrete model and compare them with those of the classical Maxwell’s equations. In Section 4, we reduce our semi-discrete model of the three-dimensional Maxwell’s equations to the two-dimensional case. Following \cite{S5, S6}, we consider the discrete Maxwell's equations on a combinatorial torus as an illustrative example and derive an explicit expression for the general solution in this setting.

  \section{Background on a discrete model}
 A detailed construction of a combinatorial model for the two-dimensional Euclidean space
$\mathbb{R}^2$ is given in \cite{S5}. In this section, we generalize that approach to the three-dimensional case.
The combinatorial model of $\mathbb{R}^3$ is defined as a three-dimensional chain complex
\begin{equation*}
C(3)=C_0(3)\oplus C_1(3)\oplus C_2(3)\oplus C_3(3)
\end{equation*}
 generated by the 0-, 1-, 2-, and 3-dimensional basis elements
 \begin{equation*}
 \{x_{k,s,m}\}, \
 \{e_{k,s,m}^1,  e_{k,s,m}^2,  e_{k,s,m}^3\}, \ \{e_{k,s,m}^{12},  e_{k,s,m}^{13},  e_{k,s,m}^{23}\}, \ \mbox{and} \ \{V_{k,s,m}\},
\end{equation*}
 respectively, where $k,  s, m \in {\mathbb Z}$.
More precisely, each basis element of $C(3)$ can be represented as the following tensor products:
\begin{align*}
x_{k,s,m}&=x_k\otimes x_s\otimes x_m, \quad V_{k,s,m}=e_k\otimes e_s\otimes e_m, \\
e_{k,s,m}^1&=e_k\otimes x_s\otimes x_m, \quad e_{k,s,m}^2=x_k\otimes e_s\otimes x_m, \quad e_{k,s,m}^3=x_k\otimes x_s\otimes e_m, \\
e_{k,s,m}^{12}&=e_k\otimes e_s\otimes x_m, \quad e_{k,s,m}^{13}=e_k\otimes x_s\otimes e_m, \quad e_{k,s,m}^{23}=x_k\otimes e_s\otimes e_m,
\end{align*}
where $x_k$ and $e_k$ are the 0- and 1-dimensional basis elements of the 1-dimensional chain complex $C$.
Geometrically, the 0-dimensional elements $x_k$ can be interpreted as points on the real line, and the 1-dimensional elements $e_k$ as open intervals between those points.
The complex $C$ thus represents a combinatorial real line, and the full complex $C(3)$ can be written as the tensor product $C(3)=C\otimes C \otimes C$.
  On the chain complex $C(3)$, we define the boundary operator  $\partial: C_r(3)\rightarrow C_{r-1}(3)$, $r=1,2,3$,
 as follows
\begin{align}\label{2.1}
\partial x_{k,s,m}&=0, \qquad  \partial e_{k,s,m}^1=x_{\tau k,s,m}-x_{k,s,m}, \nonumber \\
  \partial e_{k,s,m}^2&=x_{k, \tau s,m}-x_{k,s,m}, \qquad  \partial e_{k,s,m}^3=x_{k, s, \tau m}-x_{k,s,m},\nonumber \\
\partial e_{k,s,m}^{12}&=e_{\tau k,s,m}^2-e_{k,s,m}^2-e_{k, \tau s,m}^1+e_{ k,s,m}^1, \nonumber \\
 \partial e_{k,s,m}^{13}&=e_{\tau k,s,m}^3-e_{k,s,m}^3-e_{k,  s, \tau m}^1+e_{ k,s,m}^1,
\nonumber \\
\partial e_{k,s,m}^{23}&=e_{ k,\tau s,m}^3-e_{k,s,m}^3-e_{k,  s, \tau m}^2+e_{ k,s,m}^2,
\nonumber \\
\partial V_{k,s,m}&=e_{k,s,\tau m}^{12}-e_{k,s,m}^{12}+e_{\tau k,s,m}^{23}- e_{k,s,m}^{23}-e_{k, \tau s,m}^{13}+e_{k,s,m}^{13}.
\end{align}
Here, $\tau$
 denotes the forward shift operator, i.e., $\tau k=k+1$.
This definition extends linearly to arbitrary chains in the complex.

We now introduce the dual object to the chain complex $C(3)$, denoted by  $K(3)=K^0(3)\oplus K^1(3)\oplus K^2(3)\oplus K^3(3)$, as defined  in \cite{S5}. This dual complex has a structure analogous to that of
$C(3)$ and consists of cochains with real-valued coefficients.  Let the sets
 \begin{equation*}
 \{x^{k,s,m}\}, \
 \{e^{k,s,m}_1,  e^{k,s,m}_2,  e^{k,s,m}_3\}, \ \{e^{k,s,m}_{12},  e^{k,s,m}_{13},  e^{k,s,m}_{23}\}, \ \mbox{and} \ \{V^{k,s,m}\}
\end{equation*}
denote the basis elements of $K^0(3)$, $K^1(3)$, $K^2(3)$, and $K^2(3)$, respectively. Using these bases, cochains
 $\Phi\in K^0(3)$, $\Psi\in K^3(3)$, $A\in K^1(3)$, and $B\in K^2(3)$ can be expressed in component form as
  \begin{equation}\label{2.2}
  \Phi=\sum_{k,s,m}\Phi_{k,s,m}x^{k,s,m}, \qquad
   \Psi=\sum_{k,s,m}\Psi_{k,s,m}V^{k,s,m},
  \end{equation}
\begin{equation}\label{2.3}
A=\sum_{k,s,m}(A^1_{k,s,m}e_1^{k,s,m}+A^2_{k,s,m}e_2^{k,s,m}+A^3_{k,s,m}e_3^{k,s,m}),
\end{equation}
\begin{equation}\label{2.4}
B=\sum_{k,s,m}(B^{12}_{k,s,m}e_{12}^{k,s,m}+B^{13}_{k,s,m}e_{13}^{k,s,m}+B^{23}_{k,s,m}e_{23}^{k,s,m}),
\end{equation}
where  $\Phi_{k,s,m}, \Psi_{k,s,m},
A^i_{k,s,m}, B^{ij}_{k,s,m}\in {\mathbb R}$ for all $k, s,m \in {\mathbb Z}$ and $i,j=1,2,3$.
 Following the terminology in \cite{S5}, we refer to these cochains as forms or discrete forms.

For discrete forms \eqref{2.2}-\eqref{2.4}, the pairing with the basis elements of
$C(3)$ is defined by the following rule:
\begin{align}\label{2.5}
&\langle x_{k,s,m}, \ \Phi\rangle=\Phi_{k,s,m}, \quad \langle V_{k,s,m}, \ \Psi\rangle=\Psi_{k,s,m}, \nonumber \\
&\langle e_{k,s,m}^1, \ A\rangle=A^1_{k,s,m}, \quad
\langle e_{k,s,m}^2, \ A\rangle=A^2_{k,s,m},  \quad \langle e_{k,s,m}^3, \ A\rangle=A^3_{k,s,m}, \nonumber \\
&\langle e_{k,s,m}^{12}, \ B\rangle=B^{12}_{k,s,m}, \quad
\langle e_{k,s,m}^{13}, \ B\rangle=B^{13}_{k,s,m},  \quad \langle e_{k,s,m}^{23}, \ B\rangle=B^{23}_{k,s,m}.
\end{align}

Let $\Omega\in K^r(3)$ and let $a\in C_{r+1}(3)$  be an $(r+1)$-chain. As in \cite{S5},  the coboundary operator $d^c: K^r(3)\rightarrow K^{r+1}(3)$ is defined  through the duality relation
\begin{equation}\label{2.6}
\langle a,  \ d^c \Omega\rangle=\langle \partial a,  \ \Omega\rangle,
\end{equation}
where $\partial$ is given by \eqref{2.1}.
This operator can be regarded as a discrete analogue of the exterior derivative. Accordingly, for the forms \eqref{2.2}-\eqref{2.4}, we have
\begin{equation}\label{2.7}
d^c\Phi=\sum_{k,s,m}(\Delta_k\Phi_{k,s,m})e_1^{k,s,m}+(\Delta_s\Phi_{k,s,m})e_2^{k,s,m}+
(\Delta_m\Phi_{k,s,m})e_3^{k,s,m},
\end{equation}
\begin{align}\label{2.8}
d^cA=\sum_{k,s,m}&\big((\Delta_kA^2_{k,s,m}-\Delta_sA^1_{k,s,m})e_{12}^{k,s,m}\nonumber \\ &+(\Delta_kA^3_{k,s,m}-\Delta_mA^1_{k,s,m})e_{13}^{k,s,m} \nonumber \\
&+
(\Delta_sA^3_{k,s,m}-\Delta_mA^2_{k,s,m})e_{23}^{k,s,m}\big),
\end{align}
\begin{equation}\label{2.9}
d^cB=\sum_{k,s,m}(\Delta_kB^{23}_{k,s,m}-\Delta_sB^{13}_{k,s,m}+\Delta_mB^{12}_{k,s,m})V^{k,s,m},
\end{equation}
and we have $d^c\Psi=0$.
Here, the operators $\Delta_k, \Delta_s$, and $\Delta_m$ are finite difference operators, defined by
\begin{align*}
\Delta_k \Phi_{k,s,m} &= \Phi_{\tau k, s, m} - \Phi_{k,s,m}, \\
\Delta_s \Phi_{k,s,m} &= \Phi_{k, \tau s, m} - \Phi_{k,s,m}, \\
\Delta_m \Phi_{k,s,m} &= \Phi_{k, s, \tau m} - \Phi_{k,s,m}.
\end{align*}
Note that for any $r$-form $\Omega\in K^r(3)$, the following identity holds
\begin{equation}\label{2.10}
d^c(d^c \Omega)=0.
\end{equation}
This follows directly from  \eqref{2.1} and \eqref{2.6}.

 Finally, we extend the definitions of the $\cup$ product and the star operator, as introduced in \cite{S5}, to the 3-dimensional complex
 $K(3)$,
 For the basis elements of $K(3)$, the $\cup$ product is defined as follows
\begin{equation*}\label{}
x^{k,s,m}\cup x^{k,s,m}=x^{k,s,m}, \quad x^{k,s}\cup e^{k,s,m}_1=e^{k,s,m}_1, \quad x^{k,s,m}\cup e^{k,s,m}_2=e^{k,s,m}_2,
\end{equation*}
\begin{equation*}\label{}
x^{k,s,m}\cup e^{k,s,m}_3=e^{k,s,m}_3, \quad x^{k,s,m}\cup e^{k,s,m}_{12}=e^{k,s,m}_{12}, \quad x^{k,s,m}\cup e^{k,s,m}_{13}=e^{k,s,m}_{13},
\end{equation*}
\begin{equation*}\label{}
x^{k,s,m}\cup e^{k,s,m}_{23}=e^{k,s,m}_{23}, \ x^{k,s,m}\cup V^{k,s,m}=V^{k,s,m}, \  V^{k,s,m}\cup x^{\tau k,\tau s, \tau m}=V^{k,s,m},
\end{equation*}
\begin{equation*}\label{}
e^{k,s,m}_1\cup x^{\tau k,s,m}=e^{k,s,m}_1, \quad e^{k,s,m}_2\cup x^{k,\tau s,m}=e^{k,s,m}_2, \quad e^{k,s,m}_3\cup x^{k,s,\tau m}=e^{k,s,m}_3,
 \end{equation*}
\begin{equation*}\label{}
e^{k,s,m}_{12}\cup x^{\tau k,\tau s,m}=e^{k,s,m}_{12}, \ e^{k,s,m}_{13}\cup x^{\tau k, s,\tau m}=e^{k,s,m}_{13}, \ e^{k,s,m}_{23}\cup x^{k,\tau s,\tau m}=e^{k,s,m}_{23},
 \end{equation*}
\begin{equation*}\label{}
e^{k,s,m}_1\cup e^{\tau k,s,m}_2=e^{k,s,m}_{12}, \ e^{k,s,m}_1\cup e^{\tau k,s,m}_3=e^{k,s,m}_{13}, \
e^{k,s,m}_2\cup e^{k, \tau s,m}_1=-e^{k,s,m}_{12},
\end{equation*}
\begin{equation*}\label{}
e^{k,s,m}_2\cup e^{k, \tau s,m}_3=e^{k,s,m}_{23}, \ e^{k,s,m}_3\cup e^{k,s,\tau m}_1=-e^{k,s,m}_{13}, \
e^{k,s,m}_3\cup e^{k,s,\tau m}_2=-e^{k,s,m}_{23},
\end{equation*}
\begin{equation*}\label{}
e^{k,s,m}_1\cup e^{\tau k,s,m}_{23}=V^{k,s,m}, \ e^{k,s,m}_2\cup e^{k,\tau s,m}_{13}=-V^{k,s,m}, \
e^{k,s,m}_3\cup e^{k,s,\tau m}_{12}=V^{k,s,m},
\end{equation*}
\begin{equation*}\label{}
e^{k,s,m}_{12}\cup e^{\tau k,\tau s,m}_3=V^{k,s,m},  e^{k,s,m}_{13}\cup e^{\tau k, s,\tau m}_2=-V^{k,s,m},
e^{k,s,m}_{23}\cup e^{k,\tau s,\tau m}_1=V^{k,s,m}.
\end{equation*}
In all other cases, the product is defined to be zero.  This operation extends to arbitrary forms by linearity. As shown in \cite[Ch. 3, Proposition 2]{Dezin}, for real-valued discrete forms, the discrete analogue of the Leibniz rule holds:
\begin{equation}\label{2.11}
 d^c(\Omega\cup\Phi)=d^c \Omega\cup\Phi+(-1)^r\Omega\cup
d^c\Phi,
\end{equation}
where $r$ is the degree of $\Omega$.

The star operator $\ast: K^r(3)\rightarrow  K^{3-r}(3)$ is defined  by the rule:
\begin{align}\label{2.12}
\ast x^{k,s,m}&=V^{k,s,m}, \ \quad \ast V^{k,s,m}=x^{\tau k,\tau s,\tau m}, \nonumber \\
\ast e^{k,s,m}_1&=e^{\tau k,s,m}_{23}, \ \quad \ast e^{k,s,m}_2=-e^{k,\tau s,m}_{13}, \
\quad \ast e^{k,s,m}_3=e^{k,s,\tau m}_{12}, \nonumber \\
\ast e^{k,s, m}_{12}&=e^{\tau k,\tau s,m}_3, \quad \ast e^{k,s, m}_{13}=-e^{\tau k,s,\tau m}_2, \quad \ast e^{k,s, m}_{23}=e^{k,\tau s,\tau m}_1.
\end{align}
As before, this operation is extended to arbitrary forms by linearity. The operator $\ast$
 exhibits properties analogous to those of the Hodge star operator and can therefore be regarded as its discrete analogue.
\begin{rem}
For any discrete forms, the operation $\ast\ast$ results in a shift of all indices of the basis elements, unlike in the continuous case, where $\ast\ast A=A$ for any differential $r$-form $A$. For example, for a discrete 1-form, we have
\begin{align*}\label{2.4}
\ast\ast A&=\sum_{k,s,m}(A^1_{k,s,m}e_1^{\tau k,\tau s, \tau m}+A^2_{k,s,m}e_2^{\tau k,\tau s, \tau m}+A^3_{k,s,m}e_3^{\tau k,\tau s, \tau m})\\
&=\sum_{k,s,m}(A^1_{\sigma k, \sigma s, \sigma m}e_1^{k,s,m}+A^2_{\sigma k, \sigma s, \sigma m}e_2^{k,s,m}+A^3_{\sigma k, \sigma s, \sigma m}e_3^{k,s,m}),
\end{align*}
where  $\sigma$  denotes a unit shift to the left, i.e., $\sigma k=k-1$.
  Note that this is one of the key differences between our discrete model and the continuous case.
\end{rem}

\begin{prop}
For any $r$-form $A$ we have
\begin{equation}\label{2.13}
 d^c(\ast\ast A)=\ast\ast d^c A.
\end{equation}
\end{prop}
\begin{proof}
The proof is a direct computation. Let $A$ be a 1-form. By \eqref{2.8} and \eqref{2.12}, it follows that
\begin{align*}
d^c(\ast\ast A) &=
d^c \sum_{k,s,m} \big( A^1_{\sigma k, \sigma s, \sigma m}  e_1^{k,s,m} + A^2_{\sigma k, \sigma s, \sigma m}e_2^{k,s,m} + A^3_{\sigma k, \sigma s, \sigma m}e_3^{k,s,m} \big) \\
&= \sum_{k,s,m} \big( (\Delta_k A^2_{\sigma k, \sigma s, \sigma m} - \Delta_s A^1_{\sigma k, \sigma s, \sigma m})e_{12}^{k,s,m} \\
&\qquad\quad + (\Delta_k A^3_{\sigma k, \sigma s, \sigma m} - \Delta_m A^1_{\sigma k, \sigma s, \sigma m})e_{13}^{k,s,m} \\
&\qquad\quad + (\Delta_s A^3_{\sigma k, \sigma s, \sigma m} - \Delta_m A^2_{\sigma k, \sigma s, \sigma m})e_{23}^{k,s,m} \big ) \\
&= \sum_{k,s,m} \big( (\Delta_k A^2_{k, s, m} - \Delta_s A^1_{k, s, m})e_{12}^{\tau k,\tau s, \tau m} \\
&\qquad\quad + (\Delta_k A^3_{k, s, m} - \Delta_m A^1_{k, s, m})e_{13}^{\tau k, \tau s, \tau m} \\
&\qquad\quad + (\Delta_s A^3_{k, s, m} - \Delta_m A^2_{k, s, m})e_{23}^{\tau k, \tau s, \tau m} \big) \\
&= \sum_{k,s,m} \big( (\Delta_k A^2_{k, s, m} - \Delta_s A^1_{k, s, m})\ast\ast e_{12}^{k, s, m} \\
&\qquad\quad + (\Delta_k A^3_{k, s, m} - \Delta_m A^1_{k, s, m})\ast\ast e_{13}^{k, s, m} \\
&\qquad\quad + (\Delta_s A^3_{k, s, m} - \Delta_m A^2_{k, s, m})\ast\ast e_{23}^{k, s, m} \big) \\
&=\ast\ast d^c A.
\end{align*}
Similarly, the identity can be derived for 0-forms and 2-forms.
\end{proof}

  We define $V$  to be the three-dimensional finite chain with unit coefficients, given by
\begin{equation}\label{2.14}
  V=\sum_{k=1}^N\sum_{s=1}^S\sum_{m=1}^MV_{k,s,m}.
\end{equation}
The inner product of discrete forms over $V$  is defined as
 \begin{equation}\label{2.15}
 (\Phi, \ \Omega)_V=\langle V, \ \Phi\cup\ast\Omega\rangle,
 \end{equation}
 where  $\Phi$ and $\Omega$ are discrete forms of the same degree.
   If the forms have different degrees, the product \eqref{2.15} is defined to be zero.
  From \eqref{2.5} and \eqref{2.12}, using the definition of the $\cup$ product, we obtain the following explicit expressions.
For 0-forms or 3-forms  of the form \eqref{2.2}, the inner product becomes
   \begin{equation*}\label{}
 (\Phi, \ \Omega)_V=\sum_{k=1}^N\sum_{s=1}^S\sum_{m=1}^M\Phi_{k,s,m}\Omega_{k,s,m}.
 \end{equation*}
 For 1-forms as in \eqref{2.3}, the inner product is given by
\begin{equation*}\label{}
 (\Phi, \ \Omega)_V=\sum_{k=1}^N\sum_{s=1}^S\sum_{m=1}^M(\Phi_{k,s,m}^1\Omega_{k,s,m}^1+\Phi_{k,s,m}^2\Omega_{k,s,m}^2+\Phi_{k,s,m}^3\Omega_{k,s,m}^3)
 \end{equation*}
 and for 2-forms given by  \eqref{2.4}, it takes the form
 \begin{equation*}\label{}
 (\Phi, \ \Omega)_V=\sum_{k=1}^N\sum_{s=1}^S\sum_{m=1}^M(\Phi_{k,s,m}^{12}\Omega_{k,s,m}^{12}+\Phi_{k,s,m}^{13}\Omega_{k,s,m}^{13}+\Phi_{k,s,m}^{23}\Omega_{k,s,m}^{23}).
 \end{equation*}

The next proposition introduces the adjoint operator of $d^c$ with respect to the inner product \eqref{2.15}.
\begin{prop}
 Let $\Phi\in K^r(3)$  and $\Omega\in K^{r+1}(3)$, where $r=0,1,2$. Then the following identity holds
\begin{equation}\label{2.16}
 (d^c\Phi, \ \Omega)_V=\langle \partial V, \ \Phi\cup\ast\Omega\rangle+(\Phi, \ \delta^c\Omega)_V,
\end{equation}
 where
 \begin{equation}\label{2.17}
 \delta^c\Omega=(-1)^{r+1}\ast^{-1}d^c\ast\Omega
 \end{equation}
  and $\ast^{-1}$ denotes the inverse of the discrete Hodge star operator  $\ast$.
\end{prop}
\begin{proof}
 The proof coincides with that \cite[Proposition 2]{S5}.
\end{proof}
It is evident that the operator  $\delta^c: K^{r+1}(3) \rightarrow K^r(3)$, defined by \eqref{2.16}, serves as a discrete analogue of the codifferential $\delta$.

Note that by \eqref{2.7}, since $\ast^{-1} \ast = \mathrm{Id}$, we have
\begin{align}\label{2.18}
\ast^{-1} x^{k,s,m}&=V^{\sigma k,\sigma s,\sigma m}, \quad  \ast^{-1} V^{k,s,m}=x^{k,s, m},  \nonumber \\
\ast^{-1} e^{k,s,m}_1&=e^{k,\sigma s,\sigma m}_{23}, \quad  \ast^{-1} e^{k,s,m}_2=-e^{\sigma k,s,\sigma m}_{13},
\quad \ast^{-1} e^{k,s,m}_3=e^{\sigma k, \sigma s,m}_{12}, \nonumber \\
\ast^{-1} e^{k,s,m}_{12}&=e^{k,s,\sigma m}_3, \quad \ast^{-1} e^{k,s, m}_{13}=-e^{k,\sigma s, m}_2, \quad \ast^{-1} e^{k,s, m}_{23}=e^{\sigma k, s, m}_1.
\end{align}
 Using  these relations and \eqref{2.7},  along with  the definition of $d^c$, we can derive explicit expressions for the operator
$\delta^c$   for various types of forms.
Let $\Phi, \Psi, A$ and $B$ be the forms given by \eqref{2.2}–\eqref{2.4}. Then we obtain $\delta^c\Phi=0$, and

\begin{equation}\label{2.19}
\delta^c\Psi=\sum_{k,s,m}(\Delta_s\Psi_{k,\sigma s,m})e_{13}^{k,s,m}-(\Delta_m\Psi_{k,s,\sigma m})e_{12}^{k,s,m}-
(\Delta_k\Psi_{\sigma k,s,m})e_{23}^{k,s,m},
\end{equation}
\begin{equation}\label{2.20}
\delta^cA=\sum_{k,s,m}(-\Delta_kA^1_{\sigma k,s,m}-\Delta_sA^2_{k,\sigma s,m}-\Delta_mA^3_{k,s,\sigma m})x^{k,s,m},
\end{equation}
\begin{align}\label{2.21}
\delta^cB=\sum_{k,s,m}\big(&(\Delta_sB^{12}_{k,\sigma s,m}+\Delta_mB^{13}_{k,s,\sigma m})e_1^{k,s,m} \nonumber \\
-&(\Delta_kB^{12}_{\sigma k, s,m}-\Delta_mB^{23}_{k,s,\sigma m})e_2^{k,s,m}\nonumber \\
-&(\Delta_kB^{13}_{\sigma k, s,m}+\Delta_sB^{23}_{k,\sigma s, m})e_3^{k,s,m}\big ).
\end{align}
The  operator
\begin{equation}\label{2.22}
\Delta^c=d^c\delta^c+\delta^cd^c: K^r(3) \rightarrow K^r(3)
\end{equation}
 defines  a discrete analogue of the Laplacian on the complex $K(3)$.

  \section{Discrete Maxwell's equations in 3D}

In this section, we develop a spatial discretization framework for constructing a semi-discrete analogue of Maxwell’s equations in the three-dimensional case. The discrete model introduced in the previous section is employed to represent the spatial variables, while the temporal variable is treated continuously. Furthermore, we examine the principal properties of the resulting semi-discrete formulation and discuss its relationship with the classical Maxwell’s equations.

Let the discrete analogues of the electric and magnetic field intensities, flux densities, and the electric current and charge densities be defined by the following discrete forms:
\begin{equation}\label{3.1}
    E=\sum_{k,s,m}\left(E^1_{k,s,m}(t)e_1^{k,s,m}+E^2_{k,s,m}(t)e_2^{k,s,m}+E^3_{k,s,m}(t)e_3^{k,s,m}\right ),
\end{equation}
\begin{equation}\label{3.2}
    H=\sum_{k,s,m}\left(H^1_{k,s,m}(t)e_1^{k,s,m}+H^2_{k,s,m}(t)e_2^{k,s,m}+H^3_{k,s,m}(t)e_3^{k,s,m}\right ),
\end{equation}
\begin{equation}\label{3.3}
D=\sum_{k,s,m}\left(D^{12}_{k,s,m}(t)e_{12}^{k,s,m}+D^{13}_{k,s,m}(t)e_{13}^{k,s,m}+D^{23}_{k,s,m}(t)e_{23}^{k,s,m}\right ),
\end{equation}
\begin{equation}\label{3.4}
B=\sum_{k,s,m}\left(B^{12}_{k,s,m}(t)e_{12}^{k,s,m}+B^{13}_{k,s,m}(t)e_{13}^{k,s,m}+B^{23}_{k,s,m}(t)e_{23}^{k,s,m}\right ),
\end{equation}
\begin{equation}\label{3.5}
J=\sum_{k,s,m}\left(J^{12}_{k,s,m}(t)e_{12}^{k,s,m}+J^{13}_{k,s,m}(t)e_{13}^{k,s,m}+J^{23}_{k,s,m}(t)e_{23}^{k,s,m}\right ),
\end{equation}
\begin{equation}\label{3.6}
Q=\sum_{k,s,m}Q_{k,s,m}(t)V^{k,s,m}.
\end{equation}
For simplicity, we  omit the time variable $t$ from the components of these forms in what follows.

The semi-discrete counterparts of  Maxwell's equations \eqref{1.1}-\eqref{1.4},  with time remaining continuous,  are given by
\begin{equation}\label{3.7}
  d^c E=-\frac{dB}{dt},
 \end{equation}
 \begin{equation}\label{3.8}
 d^cH=\frac{d D}{dt}+J,
 \end{equation}
 \begin{equation}\label{3.9}
 d^cD=Q,
 \end{equation}
 \begin{equation}\label{3.10}
 d^cB=0,
 \end{equation}
 where
$d^c$  denotes the discrete exterior derivative and the discrete forms are  as defined in \eqref{3.1}–\eqref{3.5}.
Note that the time derivative  operates on the discrete two-form $B$
 (and analogously on $D$)
 as follows
\begin{equation*}
\frac{dB}{dt}=\sum_{k,s,m}\left(\frac{dB^{12}_{k,s,m}}{dt}e_{12}^{k,s,m}+\frac{dB^{13}_{k,s,m}}{dt}e_{13}^{k,s,m}+\frac{dB^{23}_{k,s,m}}{dt}e_{23}^{k,s,m}\right).
\end{equation*}
Using \eqref{2.8},  Equation \eqref{3.7} -- a semi-discrete analogue of Faraday’s law -- can be expressed in terms of difference-differential equations as follows:
\begin{align*}
\Delta_k E^2_{k,s,m} - \Delta_s E^1_{k,s,m} &= -\frac{dB^{12}_{k,s,m}}{dt}, \\
\Delta_k E^3_{k,s,m} - \Delta_m E^1_{k,s,m} &= -\frac{dB^{13}_{k,s,m}}{dt}, \\
\Delta_s E^3_{k,s,m} - \Delta_m E^2_{k,s,m} &= -\frac{dB^{23}_{k,s,m}}{dt}
\end{align*}
for all $k,s,m\in\mathbb{Z}$.

Similarly,  Equation \eqref{3.8} -- a semi-discrete analogue of Ampère’s law -- is equivalent to the following system of difference-differential equations:
\begin{align*}
\Delta_k H^2_{k,s,m} - \Delta_s H^1_{k,s,m} &= \frac{dD^{12}_{k,s,m}}{dt}  + J^{12}_{k,s,m}, \\
\Delta_k H^3_{k,s,m} - \Delta_m H^1_{k,s,m} &= \frac{dD^{13}_{k,s,m}}{dt}  + J^{13}_{k,s,m}, \\
\Delta_s H^3_{k,s,m} - \Delta_m H^2_{k,s,m} &= \frac{dD^{23}_{k,s,m}}{dt}  + J^{23}_{k,s,m}.
\end{align*}
Finally, using \eqref{2.9},  Equation \eqref{3.9}  -- a semi-discrete analogue of Gauss’ law -- and Equation \eqref{3.10} -- a discrete analog of Gauss's law for magnetism --
 can be represented as
\begin{equation*}\label{}
    \Delta_kD^{23}_{k,s,m}-\Delta_sD^{13}_{k,s,m}+\Delta_mD^{12}_{k,s,m}=Q_{k,s,m},
\end{equation*}
and
\begin{equation*}\label{}
    \Delta_kB^{23}_{k,s,m}-\Delta_sB^{13}_{k,s,m}+\Delta_mB^{12}_{k,s,m}=0.
\end{equation*}
Using \eqref{2.12},  a discrete component-wise representation of the constitutive relations \eqref{1.5} and \eqref{1.6} can be formulated as
 \begin{equation}\label{3.11}
B^{12}_{k,s,m}=\mu_0H^3_{k,s,\sigma m}, \quad B^{13}_{k,s,m}=-\mu_0H^2_{k,\sigma s, m}, \quad B^{23}_{k,s,m}=\mu_0H^1_{\sigma k,s, m},
\end{equation}
 \begin{equation}\label{3.12}
D^{12}_{k,s,m}=\varepsilon_0E^3_{k,s,\sigma m}, \quad D^{13}_{k,s,m}=-\varepsilon_0E^2_{k,\sigma s, m}, \quad D^{23}_{k,s,m}=\varepsilon_0E^1_{\sigma k,s, m}.
\end{equation}
In a similar fashion, a discrete counterpart of the dual relations \eqref{1.7} and \eqref{1.8} takes the form
\begin{equation}\label{3.13}
B^{12}_{\sigma k,\sigma s,m}=\mu_0H^3_{k,s, m}, \ B^{13}_{\sigma k,s,\sigma m}=-\mu_0H^2_{k,s, m}, \
 B^{23}_{k,\sigma s, \sigma m}=\mu_0H^1_{k,s, m},
\end{equation}
 \begin{equation}\label{3.14}
D^{12}_{\sigma k,\sigma s,m}=\varepsilon_0E^3_{k,s, m}, \ D^{13}_{\sigma k,s,\sigma m}=-\varepsilon_0E^2_{k, s, m},
 \ D^{23}_{k,\sigma s,\sigma m}=\varepsilon_0E^1_{k,s, m}.
\end{equation}
It is clear that equations \eqref{3.11} and \eqref{3.12} are not equivalent to the corresponding equations \eqref{3.13} and \eqref{3.14}, as they are in the continuous case.  This discrepancy arises from the definition of the $\ast$ operator given by \eqref{2.12} (see Remark~~2.1). In the analysis that follows, we employ both sets of equations - \eqref{3.11}, \eqref{3.12} and \eqref{3.13}, \eqref{3.14}.

Let us now present a counterpart of Poynting's theorem \eqref{1.9} in the framework of discrete forms.
\begin{prop}
The following identity holds
 \begin{equation}\label{3.15}
d^c(E\cup H)=-\frac{1}{2}\frac{d}{dt}(E\cup D+B\cup H)-E\cup J.
\end{equation}
\end{prop}
\begin{proof}
From \eqref{2.6} for the 1-form $E$, we have
\begin{equation}\label{3.16}
d^c(E\cup H)=d^cE\cup H-E\cup d^cH.
\end{equation}
By the definition of the $\cup$ product and using \eqref{2.12}, we  compute
\begin{align*}\label{}
\frac{d}{dt}(E\cup\ast E)=\frac{d}{dt}\sum_{k,s,m}\left((E^1_{k,s,m})^2+(E^2_{k,s,m})^2+(E^3_{k,s,m})^2\right)V^{k,s,m}\\
=\sum_{k,s,m}\left(\frac{d}{dt}(E^1_{k,s,m})^2+\frac{d}{dt}(E^2_{k,s,m})^2+\frac{d}{dt}(E^3_{k,s,m})^2\right)V^{k,s,m} \\
=\sum_{k,s,m}\left(2E^1_{k,s,m}\frac{dE^1_{k,s,m}}{dt}+2E^2_{k,s,m}\frac{dE^2_{k,s,m}}{dt}+2E^3_{k,s,m}\frac{dE^3_{k,s,m}}{dt}\right)V^{k,s,m}\\
=2E\cup\ast\frac{dE}{dt}=2E\cup\frac{d\ast E}{dt}.
\end{align*}
Therefore,
\begin{equation}\label{3.17}
E\cup\frac{d\ast E}{dt}=\frac{1}{2}\frac{d}{dt}(E\cup\ast E).
\end{equation}
Similarly, we obtain
\begin{equation}\label{3.18}
\frac{dB}{dt} \cup\ast B=\frac{1}{2}\frac{d}{dt}(B\cup\ast B).
\end{equation}
From the semi-discrete Maxwell's equation for the electric field \eqref{3.7}, using \eqref{3.13} and  \eqref{3.18}, it follows that
\begin{equation*}
d^cE\cup H=-\frac{dB}{dt}\cup H=-\frac{1}{\mu_0}\frac{dB}{dt}\cup\ast B=-\frac{1}{2}\frac{1}{\mu_0}\frac{d}{dt}(B\cup\ast B)=-\frac{1}{2}\frac{d}{dt}(B\cup H).
\end{equation*}
Similarly, from the semi-discrete Maxwell's equation for the magnetic field \eqref{3.8}, using \eqref{3.12}  and \eqref{3.17}, we get
\begin{align*}\label{}
E\cup d^cH&=E\cup \frac{dD}{dt}+E\cup J=\varepsilon_0E\cup \frac{d\ast E}{dt}+E\cup J \\&=\frac{1}{2}\varepsilon_0\frac{d}{dt}(E\cup\ast E)+E\cup J=\frac{1}{2}\frac{d}{dt}(E\cup D)+E\cup J.
\end{align*}
Substituting these into equation \eqref{3.16}, we obtain the desired result \eqref{3.15}.
\end{proof}
The relation \eqref{3.15} captures the conservation of electromagnetic energy in the discrete setting, mirroring the continuous Poynting theorem while being adapted to the algebraic and topological structure of discrete forms.

Using  definition of the $\cup$-product and by \eqref{2.9},  identity \eqref{3.15} can be written in component form as
\begin{align*}\label{}
&\quad\Delta_m(E^1_{k,s,m}H^2_{\tau k,s,m}-E^2_{k,s,m}H^1_{k,\tau s,m})\\
&-\Delta_s(E^1_{k,s,m}H^3_{\tau k,s,m}-E^3_{k,s,m}H^1_{k,s, \tau m}) \\
&+\Delta_k(E^2_{k,s,m}H^3_{k,\tau s,m}-E^3_{k,s,m}H^2_{k,s,\tau m})\\
&=-\frac{1}{2}\frac{d}{dt}(E^1_{k,s,m}D^{23}_{\tau k,s,m}-E^2_{k,s,m}D^{13}_{k,\tau s,m}+E^3_{k,s,m}D^{12}_{k,s,\tau m})\\
&-\frac{1}{2}\frac{d}{dt}(B^{12}_{k,s,m}H^3_{\tau k, \tau s,m}-B^{13}_{k,s,m}H^2_{\tau k, s,\tau m}+B^{23}_{k,s,m}H^1_{k, \tau s,\tau m})\\
&-E^1_{k,s,m}J^{23}_{\tau k,s,m}+E^2_{k,s,m}J^{13}_{k,\tau s,m}-E^3_{k,s,m}J^{12}_{k,s,\tau m}.
\end{align*}

Let us now consider the semi-discrete Maxwell's equations in the special case where the charge density is set to zero, i.e., $Q=0$.
Hence, Equation \eqref{3.9} becomes homogeneous.
Using \eqref{2.17} and \eqref{3.12}, we compute
\begin{equation*}
  \delta^c E=-\ast^{-1}d^c\ast E=-\frac{1}{\varepsilon_0}\ast^{-1}d^cD.
 \end{equation*}
 Since $d^cD=0$, it follows that
 \begin{equation}\label{3.19}
  \delta^c E=0.
 \end{equation}
 Applying $\delta^c$ to both sides of Equation \eqref{3.7} and using the identities \eqref{2.13}, \eqref{3.11}, we obtain
 \begin{equation*}
  \delta^cd^c E=-\frac{d(\delta^cB)}{dt}=-\mu_0\frac{d}{dt}\left(\ast^{-1}d^c\ast\ast H\right)=-\mu_0\ast\frac{d(d^cH)}{dt}.
 \end{equation*}
  Using Equation \eqref{3.8}, we then have
 \begin{equation*}
  \delta^cd^c E=-\mu_0\ast\frac{d^2D}{dt^2}-\mu_0\ast\frac{dJ}{dt}.
 \end{equation*}
 By \eqref{3.12}, this  yields
 \begin{equation*}
  \delta^cd^c E+\mu_0\varepsilon_0\frac{d^2(\ast\ast E)}{dt^2}=-\mu_0\ast\frac{dJ}{dt}.
 \end{equation*}
 Taking into account \eqref{2.22} and \eqref{3.19} this equation can be rewritten in the form
 \begin{equation}\label{3.20}
  \Delta^cE+\frac{1}{c^2}\frac{d^2(\ast\ast E)}{dt^2}=-\mu_0\ast\frac{dJ}{dt}.
 \end{equation}
 Recall that $\mu_0\varepsilon_0=\frac{1}{c^2}$, where $c$ is the speed of light in vacuum.
 Thus, Equation \eqref{3.20} represents  a semi-discrete analogue of the wave equation for the electric field.
 Equation \eqref{3.25} is equivalent  to the following system of the difference-differential equations
 \begin{align*}\label{}
&-(\Delta_k)^2E^1_{\sigma k,s,m}-(\Delta_s)^2E^1_{k,\sigma s,m}-(\Delta_m)^2E^1_{k,s,\sigma m}+\frac{1}{c^2}\frac{d^2E^1_{\sigma k, \sigma s, \sigma m}}{dt^2}\\
& \hspace{4cm}=-\mu_0\frac{dJ^{23}_{k,\sigma s,\sigma m}}{dt},\\
&-(\Delta_k)^2E^2_{\sigma k,s,m}-(\Delta_s)^2E^2_{k,\sigma s,m}-(\Delta_m)^2E^2_{k,s,\sigma m}+\frac{1}{c^2}\frac{d^2E^2_{\sigma k, \sigma s, \sigma m}}{dt^2} \\
& \hspace{4cm}=\mu_0\frac{dJ^{13}_{\sigma k, s,\sigma m}}{dt},\\
&-(\Delta_k)^2E^3_{\sigma k,s,m}-(\Delta_s)^2E^3_{k,\sigma s,m}-(\Delta_m)^2E^3_{k,s,\sigma m}+\frac{1}{c^2}\frac{d^2E^3_{\sigma k,\sigma s,\sigma m}}{dt^2}\\
& \hspace{4cm}=-\mu_0\frac{dJ^{12}_{\sigma k,\sigma s, m}}{dt},
\end{align*}
 where $(\Delta_k)^2=\Delta_k\Delta_k$.

Let us introduce the semi-discrete counterparts of the electromagnetic potentials. For reference to the continuous setting, see, for example, \cite{WR}. As in the continuous theory, a semi-discrete version of the wave equation for the discrete potentials can be derived from the semi-discrete Maxwell equations. Since the discrete magnetic flux density $B$ satisfies equation \eqref{3.10}, then by \eqref{2.10}, there is a 1-form $A$ such that
\begin{equation}\label{3.21}
  B=d^c A.
 \end{equation}
By analogy with the continuous case,  this 1-form $A$ is called the discrete magnetic vector potential. Substituting \eqref{3.21} into Equation \eqref{3.7} yields
\begin{equation*}
  d^c \left(E+\frac{dA}{dt}\right )=0.
 \end{equation*}
 It follows,  according to \eqref{2.10}, that  the discrete electric 1-form $E$ can be expressed as
 \begin{equation}\label{3.22}
 E=-d^c \Phi-\frac{dA}{dt},
 \end{equation}
where $\Phi$ is a  0-form. We interpret $\Phi$ as the discrete scalar potential.
\begin{prop}
The 1-form $E$, given by \eqref{3.22}, is invariant under the following transformation
 \begin{equation}\label{3.23}
 A'=A+d^c\Psi, \quad \Phi'=\Phi-\frac{d\Psi}{dt}.
\end{equation}
\end{prop}
\begin{proof}
Assume that
\begin{equation}\label{3.24}
E'=-d^c \Phi'-\frac{dA'}{dt}.
\end{equation}
Substituting \eqref{3.23} into \eqref{3.24} yields
\begin{equation*}
E' = -d^c \left( \Phi - \frac{d\Psi}{dt} \right) - \frac{d}{dt} \left( A + d^c \Psi \right) = -d^c \Phi + d^c \left( \frac{d\Psi}{dt} \right)- \frac{dA}{dt} - \frac{d}{dt}(d^c \Psi).
\end{equation*}
Since the time derivative and the discrete exterior derivative commute, the middle terms cancel, i.e.,
\begin{equation*}
 d^c \left( \frac{d\Psi}{dt} \right)=\frac{d}{dt}(d^c \Psi).
 \end{equation*}
Thus,
\begin{equation*}
E'=-d^c \Phi-\frac{dA}{dt}=E.
\end{equation*}
\end{proof}
The transformation \eqref{3.23} is a semi-discrete analogue of a gauge transformation.

In our hybrid discrete-continuous framework, a semi-discrete counterpart of the Lorentz gauge condition can be formulated as
\begin{equation}\label{3.25}
 -\delta^cA+\frac{1}{c^2}\frac{d\Phi}{dt}=0.
 \end{equation}
 Recall that for a 1-form $A$,  we have $\delta^cA=-\ast^{-1}d^c\ast A$.
Applying \eqref{3.12}, namely $D=\varepsilon_0\ast E$, and substituting \eqref{3.22} into Equation \eqref{3.8} we obtain
\begin{equation*}\label{}
 d^cH=\varepsilon_0\frac{d\ast E}{dt}+J=\varepsilon_0\ast\frac{d}{dt}\left(-d^c \Phi-\frac{dA}{dt}\right)+J.
 \end{equation*}
From this, applying  \eqref{3.13}, i.e., $\ast B=\mu_0 H$, and by \eqref{3.21} we have
\begin{equation*}\label{}
 d^c\ast d^c A=\mu_0\varepsilon_0\ast\left(-d^c \frac{d\Phi}{dt}-\frac{d^2A}{dt^2}\right)+\mu_0J.
 \end{equation*}
 Acting with $\ast^{-1}$ on both sides and  using the gauge condition \eqref{3.25} we obtain
 \begin{equation*}\label{}
\delta^cd^c A=-d^c\delta^c A-\mu_0\varepsilon_0\frac{d^2A}{dt^2}+\mu_0\ast^{-1}J.
 \end{equation*}
 Thus, using the notation  \eqref{2.22}, we arrive at a semi-discrete analogue of the wave equation for the potential 1-form $A$:
 \begin{equation}\label{3.26}
\Delta^c A+\frac{1}{c^2}\frac{d^2A}{dt^2}=\mu_0\ast^{-1}J.
 \end{equation}
  Using the definitions of the operators $d^c, \delta^c$, and  applying \eqref{2.18}, Equation \eqref{3.26} can be decomposed into the following system of the difference-differential equations
 \begin{align*}\label{}
&-(\Delta_k)^2A^1_{\sigma k,s,m}-(\Delta_s)^2A^1_{k,\sigma s,m}-(\Delta_m)^2A^1_{k,s,\sigma m}+\frac{1}{c^2}\frac{d^2A^1_{k,s,m}}{dt^2}=\mu_0J^{23}_{\tau k,s,m},\\
&-(\Delta_k)^2A^2_{\sigma k,s,m}-(\Delta_s)^2A^2_{k,\sigma s,m}-(\Delta_m)^2A^2_{k,s,\sigma m}+\frac{1}{c^2}\frac{d^2A^2_{k,s,m}}{dt^2}=-\mu_0J^{13}_{k, \tau s,m},\\
&-(\Delta_k)^2A^3_{\sigma k,s,m}-(\Delta_s)^2A^3_{k,\sigma s,m}-(\Delta_m)^2A^3_{k,s,\sigma m}+\frac{1}{c^2}\frac{d^2A^3_{k,s,m}}{dt^2}=\mu_0J^{12}_{k,s,\tau m}.
 \end{align*}

In the same way, we derive a semi-discrete analogue of the wave equation for the scalar potential $\Phi$. Substitution \eqref{3.12} and \eqref{3.22} into
\eqref{3.9}, we obtain
\begin{equation*}
-\varepsilon_0d^c\ast d^c \Phi -\varepsilon_0\frac{d}{dt}\left(d^c\ast A\right)=Q.
 \end{equation*}
 Applying $\ast^{-1}$ to both sides and using \eqref{2.17} along with the gauge condition \eqref{3.25}, we obtain
\begin{equation*}
\delta^cd^c \Phi +\frac{1}{c^2}\frac{d^2\Phi}{dt^2}=\frac{1}{\varepsilon_0}\ast^{-1}Q.
 \end{equation*}
Since, by definition, $\delta^c\Phi=0$, it follows that $\delta^cd^c \Phi=\Delta^c\Phi$, and thus we obtain the semi-discrete analog of the wave equation in the form
\begin{equation*}
\Delta^c\Phi +\frac{1}{c^2}\frac{d^2\Phi}{dt^2}=\frac{1}{\varepsilon_0}\ast^{-1}Q.
 \end{equation*}
 Accordingly, for any components of the forms $\Phi$ and $Q$ we have
 \begin{equation*}\label{}
-(\Delta_k)^2\Phi_{\sigma k,s,m}-(\Delta_s)^2\Phi_{k,\sigma s,m}-(\Delta_m)^2\Phi_{k,s,\sigma m}+\frac{1}{c^2}\frac{d^2\Phi_{k,s,m}}{dt^2}=\frac{1}{\varepsilon_0}Q_{k,s,m}.
\end{equation*}

\section{2D discrete Maxwell's equations on a combinatorial torus}
In this section, we reduce our semi-discrete model of the three-dimensional  Maxwell's equations to the two-dimensional case.
To this end, we adopt a combinatorial model of the two-dimensional Euclidean space
$\mathbb{R}^2$,  as detailed in \cite{S5} or \cite{S6}. As an illustrative example, we consider the semi-discrete Maxwell’s equations on a combinatorial torus and derive an explicit expression for the general solution in this setting.

  On the two-dimensional chain complex $C(2)=C\otimes C$, representing a combinatorial plane, the semi-discrete Maxwell’s equations retain the same form as in \eqref{3.7}–\eqref{3.10}. The discrete electric field intensity
$E$ remains a 1-form, expressed as
  \begin{equation*}\label{}
    E=\sum_{k,s}\left(E^1_{k,s}e_1^{k,s}+E^2_{k,s}e_2^{k,s}\right).
\end{equation*}
In this two-dimensional setting, the discrete magnetic field intensity
$H$ becomes a 0-form
\begin{equation*}\label{}
    H=\sum_{k,s}H_{k,s}x^{k,s}.
\end{equation*}
Accordingly, the discrete magnetic flux density
$B$ and the discrete charge density
$Q$ are represented as 2-forms:
\begin{equation*}\label{}
    B=\sum_{k,s}B_{k,s}V^{k,s},  \qquad   Q=\sum_{k,s}Q_{k,s}V^{k,s}.
\end{equation*}
The discrete electric  flux density field $D$ and the discrete current density $J$ become 1-forms:
\begin{equation*}\label{}
    D=\sum_{k,s}\left(D^1_{k,s}e_1^{k,s}+D^2_{k,s}e_2^{k,s}\right), \qquad J=\sum_{k,s}\left(J^1_{k,s}e_1^{k,s}+J^2_{k,s}e_2^{k,s}\right).
\end{equation*}
Following the notation in \cite{S5}, we have
\begin{equation}\label{4.1}
    d^cE=\sum_{k,s}\left(\Delta_kE^2_{k,s}-\Delta_sE^1_{k,s}\right)V^{k,s}.
\end{equation}
Then, the two-dimensional version of Equation \eqref{3.7} can be written in component form as
\begin{equation}\label{4.2}
    \Delta_kE^2_{k,s}-\Delta_sE^1_{k,s}=-\frac{dB_{k,s}}{dt}
\end{equation}
for any $k,  s \in {\mathbb Z}$. Similarly, we obtain the discrete analogue of Equation \eqref{3.9}
\begin{equation}\label{4.3}
    \Delta_kD^2_{k,s}-\Delta_sD^1_{k,s}=Q_{k,s}.
\end{equation}
 Since for the 0-form $H$ we have
 \begin{equation}\label{4.4}
     d^cH=\sum_{k,s}\left((\Delta_kH_{k,s})e^1_{k,s}+(\Delta_sH_{k,s})e^2_{k,s}\right)
\end{equation}
Equation \eqref{3.8} is equivalent to the following system of difference-differential equations
\begin{align}\label{4.5}
\Delta_k H_{k,s}&= \frac{dD^{1}_{k,s}}{dt}  + J^{1}_{k,s}, \nonumber \\
\Delta_s H_{k,s}&= \frac{dD^{2}_{k,s}}{dt}  + J^{2}_{k,s}.
\end{align}
Finally, since in the 2-dimensional case  $d^c B=0$ for any 2-form $B$, Equation \eqref{3.10} holds as an identity.

By the definition of the operation $\ast$ on the complex $K(2)$, as given in \cite{S5}, we have
\begin{equation*}\label{}
    \ast E=\sum_{k,s}\left(-E^2_{k,\sigma s}e_1^{k,s}+E^1_{\sigma k,s}e_2^{k,s}\right), \quad \ast H=\sum_{k,s}H_{k,s}V^{k,s}.
\end{equation*}
Then the  two-dimensional discrete versions of the relations \eqref{3.12} and \eqref{3.11} can be written as
\begin{align}\label{4.6}
D^1_{k,s}&= -\varepsilon_0E^2_{k,\sigma s}, \nonumber \\
D^2_{k,s}&= \varepsilon_0E^1_{\sigma k, s},
\end{align}
and
\begin{equation}\label{4.7}
     B_{k,s}= \mu_0H_{k, s}.
\end{equation}
It should be noted that in the two-dimensional model, we have
\begin{equation*}\label{}
    \ast\ast E=-\sum_{k,s}\left(E^1_{\sigma k,\sigma s}e_1^{k,s}+E^2_{\sigma k,\sigma s}e_2^{k,s}\right), \quad \ast\ast  H=\sum_{k,s}H_{\sigma k, \sigma s}V^{k,s}.
\end{equation*}
It follows immediately that for any $r$-form $A$ the following identity holds
\begin{equation}\label{4.8}
 d^c(\ast\ast A)=-\ast\ast d^c A.
\end{equation}
Compared to the three-dimensional case (see relation \eqref{2.13}), the only difference is the sign on the right-hand side.

Similarly to the previous section, we now derive a semi-discrete wave equation for the discrete electric 1-form $E$ in the two dimensional case.
From the semi-discrete Maxwell's equations, using \eqref{4.8},  we have
\begin{equation*}
  \ast^{-1}d^c\ast d^c E=\mu_0\ast\frac{d^2D}{dt^2}+\mu_0\ast\frac{dJ}{dt}.
 \end{equation*}
Using the definition of $\delta^c$ given by \eqref{2.17} and applying \eqref{4.6}, this equation can be rewritten as
 \begin{equation*}
  \delta^cd^c E=\varepsilon_0\mu_0\frac{d^2(\ast\ast E)}{dt^2}+\mu_0\ast\frac{dJ}{dt}.
 \end{equation*}
 Assuming that the 2-form $Q$ is equal to zero and using \eqref{2.22}, we then obtain the semi-discrete wave equation in the form
\begin{equation}\label{4.9}
  \Delta^cE+\frac{1}{c^2}\frac{d^2(\ast\ast E)}{dt^2}=\mu_0\ast\frac{dJ}{dt}.
 \end{equation}
 Equation \eqref{4.9} is equivalent to the following system:
 \begin{align*}\label{}
4E^1_{k,s}-E^1_{\sigma k,s}-E^1_{k,\sigma s}-E^1_{\tau k,s}-E^1_{k,\tau s}-\frac{1}{c^2}\frac{d^2E^1_{\sigma k, \sigma s}}{dt^2}
&=-\mu_0\frac{dJ^2_{k,\sigma s}}{dt},\\
4E^2_{k,s}-E^2_{\sigma k,s}-E^2_{k,\sigma s}-E^2_{\tau k,s}-E^2_{k,\tau s}-\frac{1}{c^2}\frac{d^2E^2_{\sigma k, \sigma s}}{dt^2}
&=\mu_0\frac{dJ^1_{\sigma k, s}}{dt}.
\end{align*}
Now, following \cite{S5}, let us examine the two-dimensional semi-discrete Maxwell's equations on a combinatorial torus in more detail.
To begin, we associate the basis elements of the chain complex $C(2)$ with corresponding geometric objects in $\mathbb{R}^2$. As described in \cite{S5}, consider a tiling of the plane $\mathbb{R}^2$  formed by the grid lines
$x=k$ and $y=s$, where $k,s\in\mathbb{Z}$. Each open square defined by these lines is denoted by
$V_{k,s}$, with its vertices labeled $x_{k,s}, \ x_{\tau k,s}, \ x_{k,\tau s}$, \ $x_{\tau k, \tau s}$, where $\tau k=k+1$. We define the edges
$e_{k,s}^1$ and $e_{k,s}^2$ as the open intervals  $(x_{k,s}, \ x_{\tau k,s})$ and $(x_{k,s}, \ x_{k, \tau s})$, respectively.
These geometric elements correspond directly to the combinatorial objects - that is, the basis elements of the complex $C(2)$. Next, we introduce a combinatorial torus. Recall that the torus can be regarded as the topological space obtained by taking a rectangle and identifying each pair of opposite sides with the same orientation. Let  $V$  denote the open square that corresponds to the following 2-dimensional chain
\begin{equation*}\label{}
   V=V_{1,1}+V_{2,1}+V_{1,2}+V_{2,2}.
\end{equation*}
We then identify the points and intervals on the boundary of
$V$  as follows:
\begin{align}\label{4.10}
x_{1,1}=x_{3,1}=x_{1,3}=x_{3,3}, \qquad x_{1,2}=x_{3,2}, \qquad x_{2,1}=x_{2,3}, \nonumber \\
e_{1,1}^1=e_{1,3}^1, \qquad e_{2,1}^1=e_{2,3}^1, \qquad e_{1,1}^2=e_{3,1}^2, \qquad e_{1,2}^2=e_{3,2}^2.
\end{align}
The resulting geometric object is homeomorphic to the torus. For a visual representation, see \cite[Figure~1]{S5}. Let
$C(T)$ denote the chain complex associated with this structure, referred to as a combinatorial model of the torus.  Correspondingly, let
$K(T)$ represent the cochain complex defined over $C(T)$. It is clear that the components of discrete forms defined on $C(T)$ satisfy the same conditions as in \eqref{4.10}.

On the combinatorial torus $C(T)$,  the 1-form
$E$, the 0-form $H$, and the 2-form $B$ can be expressed as
  \begin{align*}\label{}
    E&=E^1_{1,1}e_1^{1,1}+E^1_{2,1}e_1^{2,1}+E^2_{1,2}e_2^{1,2}+E^2_{1,1}e_2^{1,1} \\
    &+E^1_{1,2}e_1^{1,2}+E^1_{2,2}e_1^{2,2}+E^2_{2,2}e_2^{2,2}+E^2_{2,1}e_2^{2,1},
\end{align*}

\begin{equation*}\label{}
     H=H_{1,1}x^{1,1}+H_{2,1}x^{2,1}+H_{1,2}x^{1,2}+H_{2,2}x^{2,2},
\end{equation*}
and
\begin{equation*}\label{}
     B=B_{1,1}V^{1,1}+B_{2,1}V^{2,1}+B_{1,2}V^{1,2}+B_{2,2}V^{2,2}.
\end{equation*}
Using this notation, the discrete exterior derivatives $d^cE$ and  $d^cH$, given by \eqref{4.1} and \eqref{4.4},
 take the form
\begin{align*}\label{}
 d^cE=(E^1_{1,1}-E^1_{1,2}+E^2_{2,1}-E^2_{1,1})V^{1,1}+(E^1_{2,1}-E^1_{2,2}-E^2_{2,1}+E^2_{1,1})V^{2,1}\nonumber\\
  +(E^1_{1,2}-E^1_{1,1}+E^2_{2,2}-E^2_{1,2})V^{1,2}+(E^1_{2,2}-E^1_{2,1}+E^2_{1,2}-E^2_{2,2})V^{2,2}.
\end{align*}
\begin{align*}\label{}
  d^cH&=(H_{2,1}-H_{1,1})e_1^{1,1}+(H_{1,1}-H_{2,1})e_1^{2,1}+(H_{1,1}-H_{1,2})e_2^{1,2}\nonumber\\
  &+(H_{1,2}-H_{1,1})e_2^{1,1}+(H_{2,2}-H_{1,2})e_1^{1,2}+(H_{1,2}-H_{2,2})e_1^{2,2}
  \nonumber\\
  &+(H_{2,1}-H_{2,2})e_2^{2,2}+(H_{2,2}-H_{2,1})e_2^{2,1}.
\end{align*}
Accordingly, Equation \eqref{4.2} on $C(T)$ becomes:
\begin{align}\label{4.11}
 E^1_{1,1}-E^1_{1,2}+E^2_{2,1}-E^2_{1,1}=-\frac{dB_{1,1}}{dt}, \nonumber \\
 E^1_{2,1}-E^1_{2,2}-E^2_{2,1}+E^2_{1,1}=-\frac{dB_{2,1}}{dt}, \nonumber\\
 E^1_{1,2}-E^1_{1,1}+E^2_{2,2}-E^2_{1,2}=-\frac{dB_{1,2}}{dt}, \nonumber  \\
 E^1_{2,2}-E^1_{2,1}+E^2_{1,2}-E^2_{2,2}=-\frac{dB_{2,2}}{dt}.
\end{align}
Similarly, Equation \eqref{4.3} takes the form:
\begin{align}\label{4.12}
 D^1_{1,1}-D^1_{1,2}+D^2_{2,1}-D^2_{1,1}=Q_{1,1}, \nonumber \\
 D^1_{2,1}-D^1_{2,2}-D^2_{2,1}+D^2_{1,1}=Q_{2,1}, \nonumber \\
 D^1_{1,2}-D^1_{1,1}+D^2_{2,2}-D^2_{1,2}=Q_{1,2}, \nonumber \\
 D^1_{2,2}-D^1_{2,1}+D^2_{1,2}-D^2_{2,2}=Q_{2,2}.
\end{align}
Finally, the system \eqref{4.5} reads:
\begin{align}\label{4.13}
H_{2,1}-H_{1,1}&=\frac{dD^1_{1,1}}{dt}+J^1_{1,1}, \nonumber \\
 H_{1,1}-H_{2,1}&=\frac{dD^1_{2,1}}{dt}+J^1_{2,1}, \nonumber \\
 H_{1,1}-H_{1,2}&=\frac{dD^2_{1,2}}{dt}+J^2_{1,2}, \nonumber \\
H_{1,2}-H_{1,1}&=\frac{dD^2_{1,1}}{dt}+J^2_{1,1}, \nonumber \\
H_{2,2}-H_{1,2}&=\frac{dD^1_{1,2}}{dt}+J^1_{1,2}, \nonumber \\
H_{1,2}-H_{2,2}&=\frac{dD^1_{2,2}}{dt}+J^1_{2,2}, \nonumber \\
H_{2,1}-H_{2,2}&=\frac{dD^2_{2,2}}{dt}+J^2_{2,2}, \nonumber \\
H_{2,2}-H_{2,1}&=\frac{dD^2_{2,1}}{dt}+J^2_{2,1}.
\end{align}
Thus, Equations \eqref{4.11}-\eqref{4.13} represent a semi-discrete counterpart of Maxwell's equations on the combinatorial torus.
According to \eqref{4.10} the relations \eqref{4.6} and \eqref{4.7} become
\begin{align}\label{4.14}
D^1_{1,1}&= -\varepsilon_0E^2_{1,2}, \quad D^1_{2,1}= -\varepsilon_0E^2_{2,2}, \quad  D^1_{1,2}= -\varepsilon_0E^2_{1,1}, \quad D^1_{2,2}= -\varepsilon_0E^2_{2,1},\nonumber \\
D^2_{1,1}&= \varepsilon_0E^1_{2, 1}, \quad D^2_{2,1}= \varepsilon_0E^1_{1,1}, \quad D^2_{1,2}= \varepsilon_0E^1_{2, 2}, \quad
D^2_{2,2}= \varepsilon_0E^1_{1, 2},
\end{align}
and
\begin{equation}\label{4.15}
     B_{1,1}= \mu_0H_{1, 1}, \quad  B_{2,1}= \mu_0H_{2, 1}, \quad  B_{1,2}= \mu_0H_{1, 2}, \quad  B_{2,2}= \mu_0H_{2,2}.
\end{equation}

A natural question in this framework is whether the system of semi-discrete Maxwell equations on the combinatorial torus is solvable. The following discussion addresses this question. For simplicity, we adopt natural units in which the fundamental constants $\mu_0$ and $\varepsilon_0$ are set to 1. We also assume that $Q = 0$ and $J = 0$, meaning that we are considering a region free of charges and currents.
Under these assumptions, and using relations \eqref{4.14} and \eqref{4.15}, Equations \eqref{4.11} and \eqref{4.13} reduce to the following system of linear homogeneous ordinary differential equations:
\begin{align}\label{4.16}
\frac{dE^1_{1,1}}{dt}&=H_{2,2}-H_{2,1}, \nonumber \\
 \frac{dE^1_{2,1}}{dt}&=H_{1,2}-H_{1,1}, \nonumber \\
\frac{dE^2_{1,2}}{dt}&=-H_{2,1}+H_{1,1}, \nonumber \\
\frac{dE^2_{1,1}}{dt}&=-H_{2,2}+H_{1,2}, \nonumber \\
\frac{dE^1_{1,2}}{dt}&=H_{2,1}-H_{2,2}, \nonumber \\
\frac{dE^1_{2,2}}{dt}&=H_{1,1}-H_{1,2}, \nonumber \\
\frac{dE^2_{2,2}}{dt}&=-H_{1,1}+H_{2,1}, \nonumber \\
\frac{dE^2_{2,1}}{dt}&=-H_{1,2}+H_{2,2}, \nonumber \\
\frac{dH_{1,1}}{dt}&=-E^1_{1,1}+E^1_{1,2}-E^2_{2,1}+E^2_{1,1}, \nonumber \\
\frac{dH_{2,1}}{dt}&=-E^1_{2,1}+E^1_{2,2}+E^2_{2,1}-E^2_{1,1}, \nonumber \\
\frac{dH_{1,2}}{dt}&=-E^1_{1,2}+E^1_{1,1}-E^2_{2,2}+E^2_{1,2}, \nonumber \\
\frac{dH_{2,2}}{dt}&=-E^1_{2,2}+E^1_{2,1}-E^2_{1,2}+E^2_{2,2}.
\end{align}
Similarly, the system \eqref{4.12} becomes:
\begin{align}\label{4.17}
E^1_{1,1}-E^1_{2,1}-E^2_{1,2}+E^2_{1,1}=0, \nonumber \\
 E^1_{2,1}-E^1_{1,1}-E^2_{2,2}+E^2_{2,1}=0, \nonumber \\
 E^1_{1,2}-E^1_{2,2}-E^2_{1,1}+E^2_{1,2}=0, \nonumber \\
 E^1_{2,2}-E^1_{1,2}+E^2_{2,2}-E^2_{2,1}=0.
\end{align}
We proceed as in \cite{S5} and present a matrix form of Equations \eqref{4.16} and \eqref{4.17}.
 Let us introduce the following row vectors:
\begin{equation*}\label{}
 [H]=[H_{1,1} \ H_{2,1} \ H_{1,2} \ H_{2,2}], \quad
 [E]=[E^1_{1,1} \ E^1_{2,1} \ E^2_{1,2} \ E^2_{1,1} \ E^1_{1,2} \ E^1_{2,2} \ E^2_{2,2} \ E^2_{2,1}],
 \end{equation*}
 \begin{equation*}\label{}
  [EH]=
   \begin{bmatrix} E^1_{1,1} \ E^1_{2,1} \ E^2_{1,2} \ E^2_{1,1} \ E^1_{1,2} \ E^1_{2,2} \ E^2_{2,2} \  E^2_{2,1} \ H_{1,1} \ H_{2,1}\ H_{1,2} \ H_{2,2}\end{bmatrix}.
  \end{equation*}
  Denote by $[\cdot]^\top$ the corresponding column vector. Using this notation, Equations \eqref{4.16} and \eqref{4.17}  can be rewritten as
 \begin{equation}\label{4.18}
\frac{d}{dt}[EH]^\top=\mathcal{M}\cdot[EH]^\top,
 \end{equation}
 and
  \begin{equation}\label{4.19}
\mathcal{M}_1\cdot[EH]^\top=[0]^\top,
 \end{equation}
 respectively, where
 \begin{equation*}
 \mathcal{M}=
\begin{bmatrix}
0 & 0 & 0 & 0 & 0 & 0 & 0 & 0 & 0 & -1 & 0 & 1 \\
0 & 0 & 0 & 0 & 0 & 0 & 0 & 0 & -1 & 0 & 1 & 0 \\
0 & 0 & 0 & 0 & 0 & 0 & 0 & 0 & 1 & -1 & 0 & 0 \\
0 & 0 & 0 & 0 & 0 & 0 & 0 & 0 & 0 & 0 & 1 & -1 \\
0 & 0 & 0 & 0 & 0 & 0 & 0 & 0 & 0 & 1 & 0 & -1 \\
0 & 0 & 0 & 0 & 0 & 0 & 0 & 0 & 1 & 0 & -1 & 0 \\
0 & 0 & 0 & 0 & 0 & 0 & 0 & 0 & -1 & 1 & 0 & 0 \\
0 & 0 & 0 & 0 & 0 & 0 & 0 & 0 & 0 & 0 & -1 & 1 \\
-1 & 0 & 0 & 1 & 1 & 0 & 0 & -1 & 0 & 0 & 0 & 0 \\
0 & -1 & 0 & -1 & 0 & 1 & 0 & 1 & 0 & 0 & 0 & 0 \\
1 & 0 & 1 & 0 & -1 & 0 & -1 & 0 & 0 & 0 & 0 & 0 \\
0 & 1 & -1 & 0 & 0 & -1 & 1 & 0 & 0 & 0 & 0 & 0
\end{bmatrix}
\end{equation*}
and
\begin{equation*}
\mathcal{M}_1=
\begin{bmatrix}
1 & -1 & -1 & 1 & 0 & 0 & 0 & 0 \\
-1 & 1 & 0 & 0 & 0 & 0 & -1 & 1\\
0 & 0 & 1 & -1 & 1 & -1 & 0 & 0\\
0 & 0 & 0 & 0 & -1 & 1 & 1 & -1
\end{bmatrix}.
\end{equation*}
By applying row reduction, we obtain the row echelon form of $\mathcal{M}_1$:
\begin{equation*}
\begin{bmatrix}
1 & -1 & 0 & 0 & 0 & 0 & 1 & -1 \\
0 & 0 & 1 & -1 & 0 & 0 & 1 & -1\\
0 & 0 & 0 & 0 & 1 & -1 & -1 & 1\\
0 & 0 & 0 & 0 & 0 & 0 & 0 & 0
\end{bmatrix}.
\end{equation*}
Hence, the matrix $\mathcal{M}_1$ has rank 3. It follows that a solution of Equation  \eqref{4.19} (or, equivalently, the system \eqref{4.17}) can be expressed as
\begin{align}\label{4.20}
E^1_{1,1}=E^1_{2,1}-E^2_{2,2}+E^2_{2,1}, \nonumber \\
 E^2_{1,2}=E^2_{1,1}-E^2_{2,2}+E^2_{2,1}, \nonumber \\
 E^1_{1,2}=E^1_{2,2}+E^2_{2,2}-E^2_{2,1},
\end{align}
where the variables $E^1_{2,1}$, $E^2_{1,1}$, $E^1_{2,2}$, $E^2_{2,2}$, and $E^2_{2,1}$ can be chosen arbitrarily.
Under condition \eqref{4.20} the system  \eqref{4.18} reduces to the following system of nine equations:
\begin{equation}\label{4.21}
\frac{d}{dt}[\tilde{E}H]^\top=\mathcal{M}_2\cdot[\tilde{E}H]^\top,
 \end{equation}
where
\begin{equation*}\label{}
 [\tilde{E}H]=
 \begin{bmatrix}
E^1_{2,1} & E^2_{1,1} & E^1_{2,2} & E^2_{2,2} & E^2_{2,1} & H_{1,1} & H_{2,1} & H_{1,2} & H_{2,2}
 \end{bmatrix}
\end{equation*}
and
\begin{equation*}
 \mathcal{M}_2=
\begin{bmatrix}
0 & 0 & 0 & 0 & 0 & -1 & 0 & 1 & 0 \\
0 & 0 & 0 & 0 & 0 & 0 & 0 & 1 & -1 \\
0 & 0 & 0 & 0 & 0 & 1 & 0 & -1 & 0 \\
0 & 0 & 0 & 0 & 0 & -1 & 1 & 0 & 0 \\
0 & 0 & 0 & 0 & 0 & 0 & 0 & -1 & 1 \\
-1 & 1 & 1 & 2 & -3 & 0 & 0 & 0 & 0 \\
-1 & -1 & 1 & 0 & 1 & 0 & 0 & 0 & 0 \\
1 & 1 & -1 & -4 & 3 & 0 & 0 & 0 & 0 \\
1 & -1 & -1 & 2 & -1 & 0 & 0 & 0 & 0
\end{bmatrix}.
\end{equation*}
By direct computation, the characteristic polynomial $\chi(\lambda)$ of the matrix $\mathcal{M}_3$  is found to be
\begin{equation*}
\chi(\lambda) = -\lambda^3 (\lambda - 2)^2 (\lambda + 2)^2 (\lambda^2 + 8).
\end{equation*}
This factorization reveals the eigenvalues of $\mathcal{M}_3$ as follows: $\lambda = 0$ with multiplicity 3; $\lambda = \pm 2$ , each with multiplicity 2; and $\lambda = \pm 2\sqrt{2}i$, each with multiplicity 1.
Accordingly, we can compute eigenvectors for each  eigenvalues. The following three eigenvectors correspond to $\lambda = 0$
 \begin{align*}\label{}
 \mathbf{h}_1 = \begin{bmatrix} 1 & 0 & 1 & 0 & 0 & 0 & 0 & 0 & 0 \end{bmatrix}^\top,\\
\mathbf{h}_2 = \begin{bmatrix} 0 & 1 & 0 & 1 & 1 & 0 & 0 & 0 & 0 \end{bmatrix}^\top, \\
\mathbf{h}_3 = \begin{bmatrix} 0 & 0 & 0 & 0 & 0 & 1 & 1 & 1 & 1 \end{bmatrix}^\top.
 \end{align*}
For $\lambda = -2$ the corresponding eigenvectors are
 \begin{align*}\label{}
 \mathbf{h}_4 &= \begin{bmatrix} -\frac{1}{2} & -\frac{1}{2} & \frac{1}{2} & \frac{1}{2} & \frac{1}{2} & 0 & -1 & 1 & 0 \end{bmatrix}^\top, \\
\mathbf{h}_5 &= \begin{bmatrix} -\frac{1}{2} & \frac{1}{2} & \frac{1}{2} & -\frac{1}{2} & -\frac{1}{2} & -1 & 0 & 0 & 0 \end{bmatrix}^\top,
\end{align*}
 and for $\lambda = 2$, the eigenvectors are
  \begin{align*}
 \mathbf{h}_6& = \begin{bmatrix} \frac{1}{2} & \frac{1}{2} & -\frac{1}{2} & -\frac{1}{2} & -\frac{1}{2} & 0 & -1 & 1 & 0 \end{bmatrix}^\top, \\
 \mathbf{h}_7& = \begin{bmatrix} \frac{1}{2} & -\frac{1}{2} & -\frac{1}{2} & \frac{1}{2} & \frac{1}{2} & -1 & 0 & 0 & 1 \end{bmatrix}^\top.
\end{align*}
 respectively. Finally, for the complex eigenvalues $\lambda = \pm 2\sqrt{2}i$ the corresponding eigenvectors are given by $\mathbf{h}_8\pm i\mathbf{h}_9$, where
\begin{align*}
 \mathbf{h}_8& = \begin{bmatrix} 0 & 0 & 0 & 0 & 0 & 1 & -1 & -1 & 1 \end{bmatrix}^\top, \\
 \mathbf{h}_9& = \begin{bmatrix} -\frac{\sqrt{2}}{2} & -\frac{\sqrt{2}}{2} & \frac{\sqrt{2}}{2} & -\frac{\sqrt{2}}{2} & \frac{\sqrt{2}}{2} & 0 & 0 & 0 & 0 \end{bmatrix}^\top.
\end{align*}
Thus, the general solution of  equation \eqref{4.21} can be written  as
\begin{align*}
[\tilde{E}H]^\top &= c_1 \mathbf{h}_1 + c_2 \mathbf{h}_2 + c_3 \mathbf{h}_3
+ c_4 \mathbf{h}_4 e^{-2t} + c_5 \mathbf{h}_5 e^{-2t}
+ c_6 \mathbf{h}_6 e^{2t} + c_7 \mathbf{h}_7 e^{2t} \\
&\quad + c_8 \mathbf{h}_8 \cos(2\sqrt{2}t) + c_9 \mathbf{h}_9 \sin(2\sqrt{2}t),
\end{align*}
 where $c_i\in \mathbb{R}$ are arbitrary constants.
 This expression, together with the representation \eqref{4.20}, yields the general solution of the system of semi-discrete Maxwell equations \eqref{4.18} on the combinatorial torus.

 \end{document}